\documentclass[11pt]{article}

\usepackage{amsmath,fullpage}
\usepackage{xspace}
\usepackage{amssymb}
\usepackage{epsfig}

\newtheorem{Thm}{Theorem}
\newtheorem{Lem}{Lemma}

\newtheorem{Def}{Definition}

\newenvironment{proof}{\noindent {\textbf{Proof }}}{$\Box$
\medskip}

\newcommand{\defeq}{\stackrel{\mathsf{def}}{=}}

\newcommand\bR{\mbox{$\mathbb{R}$}}
\newcommand\X{\mathcal{X}}
\newcommand\Y{\mathcal{Y}}
\newcommand\Z{\mathcal{Z}}
\newcommand\mcA{\mathcal{A}}

\newcommand\mcP{\mathcal{P}}

\newcommand{\oprt}{{\mathsf{prt}}}
\newcommand{\prt}{{\mathsf{pprt}}}

\newcommand{\R}{{\mathsf{R}}}

\newcommand{\pub}{{\mathsf{pub}}}

\newcommand\ve{\epsilon}

\mathchardef\mhyphen="2D

\newcommand{\suppress}[1]{}
\newcommand\COMMENT[1]{}

\newcommand\abs[1]{| #1 |}

\usepackage{fancyhdr}
\begin{document}
\title{A quadratically tight partition bound for classical communication complexity and query complexity}
\author{Rahul Jain\thanks{Centre for Quantum Technologies and Department of Computer Science, National University of Singapore. Email: {\tt rahul@comp.nus.edu.sg}} \quad Troy Lee\thanks{Nanyang Technological University and 
Center for Quantum Technologies. Email: {\tt troyjlee@gmail.com}}  \quad Nisheeth K. Vishnoi\thanks{Microsoft Research, India. Email: {\tt nisheeth.vishnoi@gmail.com}} }
\date{}
\maketitle

\begin{abstract}
In this work we introduce, both for classical communication complexity and query complexity,  a modification of the {\em partition bound} introduced by Jain and Klauck~\cite{JainK10}. We call it the {\em public-coin partition bound}. We show that (the logarithm to the base two of) its communication complexity and query complexity versions form, for all relations, a quadratically tight lower bound on the public-coin randomized communication complexity and randomized query complexity respectively.  
\end{abstract}


\section{Introduction}

The {\em partition bound} introduced by Jain and Klauck~\cite{JainK10} is known to be one of the strongest lower bound methods in classical communication complexity and query complexity. It is known to be stronger than most other lower bound methods, both in communication complexity and query complexity, except its relationship with the {\em information complexity} lower bound method in communication complexity is unknown. It is an interesting open question, in both these settings, as to how tight this lower bound method is. We are not aware, to the best of our knowledge, of any function or relation where this method is asymptotically weaker either for communication complexity or for query complexity. 

 In this work we introduce, both for communication complexity and query complexity, a modification of the partition bound which we call the {\em public-coin partition bound}. Analogous to the partition bound, our new bound  is also a linear-programming based lower bound method. We show that (the logarithm to the base two of) its communication and query complexity versions continue to form a lower bound on the public-coin communication complexity and randomized query complexity respectively.  In addition we show that the square of  (the logarithm to the base two of)  its communication and query complexity versions form an upper bound on the public-coin communication complexity and randomized query complexity respectively. Also it is easily seen via their linear programs that our new bound is stronger than the partition bound for all relations, both in communication complexity and query complexity.

\section{Communication complexity}
In this section we introduce our new bound in the communication complexity setting.  Let us first recall the partition bound of~\cite{JainK10}.
\begin{Def}[Partition bound~\cite{JainK10}]
Let $f \subseteq \X \times \Y \times \Z$ be a relation. Let $\ve > 0$. The $\ve$-partition bound of $f$, denoted 
$\oprt_\ve(f)$, is given by the  optimal value of the following linear program. Below $R$ represents a rectangle in 
$\X \times \Y$ and $(x,y,z) \in \X \times \Y \times \Z$.

\vspace{0.2in}

{\footnotesize
\begin{minipage}{3in}
    \centerline{ \underline{Primal}}\vspace{-0.2in}
 \begin{align*}
      \text{min:}\quad & \sum_{z} \sum_{R}  w_{z,R} \\
      & \forall (x,y) : \sum_{z: (x,y,z) \in f} \sum_{R: (x,y) \in R} w_{z,R} \geq 1 - \epsilon,\\
      & \forall (x,y) : \sum_{R: (x,y) \in R} \quad \sum_{z}  w_{z,R} = 1 , \\
      & \forall (z,R)  : w_{z,R} \geq 0 \enspace .
    \end{align*}
\end{minipage}
\begin{minipage}{3in}\vspace{-0.3in}
    \centerline{\underline{Dual}}\vspace{-0.2in}
 \begin{align*}
      \text{max:}\quad &   (1-\epsilon)\sum_{(x,y)} \mu_{x,y} + \sum_{(x,y)}\phi_{x,y} \\
      &  \forall (z,R) : \sum_{(x,y)\in R : (x,y,z) \in f} \mu_{x,y}   + \sum_{(x,y)\in  R} \phi_{x,y}  \leq 1,\\
	& \forall (x,y) : \mu_{x,y} \geq 0, \phi_{x,y} \in \bR \enspace .
 \end{align*}
\end{minipage}
}
\end{Def}

Our new bound is defined as follows. 
\begin{Def}[Public-coin partition bound]
Let $f \subseteq \X \times \Y \times \Z$ be a relation. Let $\ve > 0$. The $\ve$-public-coin partition bound of $f$, denoted 
$\prt_\ve(f)$, is given by the optimal value of the following linear program. Below $R$ represents a rectangle in $\X \times \Y$ and $P$ represents a partition along with outputs in $\Z$; that is $P = \{(z_1, R_1), (z_2, R_2), \cdots, (z_m, R_m)  \}$, such that $\{R_1, \cdots, R_m \}$ form a partition of $\X \times \Y$ into rectangles and $\forall i \in [m], z_i \in \Z$.

\vspace{0.2in}

{\footnotesize
\begin{minipage}{3in}
    \centerline{ \underline{Primal}}\vspace{-0.2in}
 \begin{align*}
      \text{min:}\quad & \sum_{z} \sum_{R}  w_{z,R} \\
      & \forall (x,y) : \sum_{z: (x,y,z) \in f} \sum_{R: (x,y) \in R} w_{z,R} \geq 1 - \epsilon,\\
      & \forall (x,y) : \sum_{R: (x,y) \in R} \quad \sum_{z}  w_{z,R} = 1 , \\
      & \forall (z,R) : w_{z,R} = \sum_{P: (z,R) \in P} a_p , \\
      & \sum_P a_P = 1  ,\\
      & \forall (z,R)  : w_{z,R} \geq 0 ; \quad  \forall P : a_P \geq 0 \enspace .
    \end{align*}
\end{minipage}
\begin{minipage}{3in}\vspace{-0.5in}
    \centerline{\underline{Dual}}\vspace{-0.2in}
 \begin{align*}
      \text{max:}\quad &   (1-\epsilon)\sum_{(x,y)} \mu_{x,y} + \sum_{(x,y)}\phi_{x,y} + \lambda \\
      &  \forall (z,R) : \sum_{(x,y)\in R : (x,y,z) \in f} \mu_{x,y}   + \sum_{(x,y)\in  R} \phi_{x,y}  + v_{z,R} \leq 1,\\
       &  \forall P : \sum_{(z,R)\in P} v_{z,R}    \geq \lambda ,\\
	& \forall (x,y) : \mu_{x,y} \geq 0, \phi_{x,y} \in \bR  ; \quad \forall (z,R) : v_{z,R} \in \bR , \\
     & \lambda \in \bR \enspace .
 \end{align*}
\end{minipage}
}
\end{Def}

We show that (the logarithm to the base two of) it is a lower bound on public-coin randomized communication complexity (please refer to~\cite{kushilevitz&nisan:cc} for standard definitions in communication complexity).

\begin{Lem}\label{lem:lowercc}
Let $f \subseteq \X \times \Y \times \Z$ be a relation. Let $\ve > 0$. Let $\R^{\pub}_{\ve} (f)$ represents the public-coin communication complexity of $f$ with worst-case error $\ve$. Then,
$$\log_2 \prt_{\ve} (f) \leq \R^{\pub}_{\ve} (f) .$$
\end{Lem}
\begin{proof} This proof goes along similar lines as the proof of~\cite{JainK10} for analogous result about the partition bound.
 
Let $\mcP$ be a public coin randomized protocol for $f$ with communication $c \defeq \R^\pub_\epsilon(f)$ and worst case error $\epsilon$. For binary string $r$, let $\mcP_r$ represent the deterministic communication protocol obtained from $\mcP$ on fixing the public coins to $r$.  Every deterministic communication protocol amounts to partitioning the inputs in $\X\times\Y$ into rectangles and outputting an element in $\Z$ corresponding to each rectangle in the partition.  Let $P_r = \{(z^r_1, R^r_1), (z^r_2, R^r_2), \cdots, (z^r_m, R^r_m)  \}$, be the corresponding partition along with the outputs, that is $\{R^r_1, \cdots, R^r_m \}$ form a partition of $\X \times \Y$ into rectangles and $\forall i \in [m], z^r_i \in \Z$. Let $q_r$ represent the probability of string $r$ in $\mcP$. For $P_r$ define $a'_{P_r} \defeq q_r$. For the partitions $P$ that do not correspond to any random string $r$ in $\mcP$, define $a'_P=0$.  For any $(z,R)$ define,
$$w'_{z,R} \defeq  \sum_{P : (z,R) \in P} a'_{P} \enspace .$$
It is easily seen that for all $(x,y,z) \in \X \times\Y\times\Z$:
$$\Pr[\mcP \mbox{ outputs } z \mbox{ on input } (x,y)] = \sum_{R: (x,y)\in R} w'_{z,R} \enspace .$$
Since the protocol has error at most $\epsilon$ on all inputs we get the constraints:
$$\forall (x,y) : \sum_{z: (x,y,z) \in f} \sum_{R: (x,y) \in R} w'_{z,R} \geq 1 - \epsilon .$$
Also since the $\Pr[\mcP \mbox{ outputs some } z\in\Z \mbox{ on input } (x,y)]= 1$, we get the constraints:
$$\forall (x,y) :  \sum_{z}  \sum_{R: (x,y) \in R} w'_{z,R} = 1 \enspace .$$
We also have by construction:
$$\sum_{P} a'_{P} = 1 ; \quad \forall (z,R): w'_{z,R} \geq 0; \quad \forall P: a'_P \geq 0 .$$
Therefore   $\{w'_{z,R}\} \cup \{a'_P \}$ is feasible for the primal of $\prt_\epsilon(f)$. 

We know that for each $r$,  $\abs{P_r} \leq 2^c$, since the communication in $\mcP_r$ is at most $c$ bits. Hence,
$$\prt_\epsilon(f) \leq \sum_{z} \sum_{R }  w'_{z,R} = \sum_r a'_{P_r} \cdot |P_r| \leq 2^c \sum_r a'_{P_r} = 2^c \enspace .$$
\end{proof}

Next we show that the square of (the logarithm to the base two of) our new bound forms an upper bound on the public-coin communication complexity.

\begin{Thm}
\label{thm:maincc}Let $f \subseteq \X \times \Y \times \Z$ be a relation. Let $\ve > 0$. We have,
$$ \R^{\pub}_{2\ve} (f) \leq \left(\log_2 \prt_\ve(f) + \log_2 \frac{1}{\ve} + 1\right)^2 . $$  
\end{Thm}
\begin{proof}
Let $\prt_\ve(f)  = 2^c $. Let $\{ w_{z,R} \} \cup \{a_P\}$ be an optimal solution for the primal. Let $n_P$ be the number of rectangles in $P$. We have,
$$ \sum_P a_P \cdot n_P = \sum_{z,R} w_{z,R} = 2^c .$$
Define $B \defeq \{ P ~| ~ n_P \geq \frac{1}{\ve} 2^c\}$. Then $\delta \defeq \sum_{P \in B} a_P \leq \ve$. Define $a'_P \defeq \frac{1}{1-\delta} a_P $ for $P\notin B$ and $a'_P \defeq 0$ for $P \in B$. Define $w'_{z,R} \defeq \sum_{P: (z,R) \in P} a'_P$. Then we have,
\begin{align*}
      & \forall (x,y) : \sum_{z: (x,y,z) \in f} \sum_{R: (x,y) \in R} w'_{z,R} \geq 1 - 2\epsilon,\\
      & \forall (x,y) : \sum_{R: (x,y) \in R} \quad \sum_{z}  w'_{z,R} = 1 , \\
      & \forall (z,R) : w'_{z,R} = \sum_{P: (z,R) \in P} a'_p , \\
      & \sum_P a'_P = 1  ,\\
      & \forall (z,R)  : w'_{z,R} \geq 0 ; \quad \forall P : a'_P \geq 0 \enspace .
\end{align*}

We know that a partition with $m$ rectangles can be realized by a communication protocol with communication $(\lceil \log_2 m \rceil)^2$ (arguments as in the proof of Theorem 2.11 of~\cite{kushilevitz&nisan:cc}, we reproduce them in Section~\ref{sec:partcc} for completeness). Consider a public-coin communication protocol $\Pi$ as follows. 
\begin{enumerate}
\item Alice and Bob (using public coins) choose a $P= \{(z_1, R_1), (z_2, R_2), \cdots, (z_m, R_m) \} $ with probability $a_P'$.
\item They communicate to realize the partition $\{R_1, R_2, \cdots, R_m\}$ with communication bounded by $(c + \log_2 \frac{1}{\ve} + 1)^2$.
\item If they end up with rectangle $R_i$, they output $z_i$.
\end{enumerate}
It is clear that the worst case communication of the protocol is bounded by $(c + \log_2 \frac{1}{\ve} + 1)^2$. The condition,
 $$\forall (x,y) : \sum_{z: (x,y,z) \in f} \sum_{R: (x,y) \in R} w'_{z,R} \geq 1 - 2\epsilon ,$$
implies that the protocol has worst case error at most $2 \ve$.
Therefore,
$$ \R^{\pub}_{2\ve} (f) \leq \left(\log_2 \prt_\ve(f) + \log_2 \frac{1}{\ve}  + 1\right)^2 .  $$

\end{proof}

\section{Query complexity}
In this section we introduce our new bound in the query complexity setting.

Let $f \subseteq \{0,1\}^n \times \Z$ be a relation.  An {\em assignment} $A : S \rightarrow \{0,1\}^l$ is an assignment of values to some subset $S$ of $n$ variables (with $|S| = l$). We say that $A$ is {\em consistent} with $x\in\{0,1\}^n$ if $x_i= A(i)$ for all $i\in S$. We write $x\in A$ as shorthand for `$A$ is consistent with $x$'. We write $|A|$ to represent the size of $A$ which is the cardinality of $S$ (not to be confused with the number of consistent inputs). Furthermore we say that an index $i$ {\em appears} in $A$, iff $i \in S$ where $S$ is the subset of $[n]$ corresponding to $A$.
Let $\mcA$ denote the set of all assignments.  Below we assume $x \in \{0,1\}^n $, $A \in \mcA$ and $z \in \Z$. Below $P$ represents a partition along with outputs in $\Z$; that is $P = \{(z_1, A_1), (z_2, A_2), \cdots, (z_m, A_m)  \}$, such that $\{A_1, \cdots, A_m \}$ form a partition of $\{0,1\}^n$  into assignments (that is for each $x \in \{0,1\}^n$, there is a unique $i\in [m]$ such that $x \in A_i$) and $\forall i \in [m], z_i \in \Z$.

Let us first recall the partition bound of~\cite{JainK10}.
\begin{Def}[Partition bound~\cite{JainK10}]

Let $f \subseteq  \{0,1\}^n \times \Z$ be a relation. Let $\ve > 0$. The $\ve$-partition bound of $f$, denoted $\oprt_\ve(f)$, is given by the optimal value of the following linear program.

\vspace{0.2in}

{\footnotesize
\begin{minipage}{3in}
    \centerline{\underline{Primal}}\vspace{-0.1in}
    \begin{align*}
      \text{min:}\quad & \sum_{z} \sum_{A}  w_{z,A} \cdot 2^{|A|} \\
       \quad &  \forall x  : \sum_{z: (x,z) \in f} \sum_{A: x \in A} w_{z,A} \geq 1 - \epsilon,\\
      & \forall x : \sum_{A: x \in A} \quad \sum_{z}  w_{z,A} = 1 , \\
      & \forall (z,A)  : w_{z,A} \geq 0 \enspace .
         \end{align*}
\end{minipage}
\begin{minipage}{3in}\vspace{-0.3in}
    \centerline{\underline{Dual}}\vspace{-0.1in}
        \begin{align*}
    \text{max:}\quad &   (1-\epsilon)\sum_{x} \mu_{x} + \sum_{x}\phi_{x} \\
      &  \forall (z,A) : \sum_{x\in A : (x,z) \in f} \mu_{x}   + \sum_{x\in  A} \phi_{x}   \leq 2^{|A|},\\
 	& \forall x : \mu_{x} \geq 0, \phi_{x} \in \bR  \enspace .
    \end{align*}
\end{minipage}
}
\end{Def}
Our new bound is defined as follows.
\begin{Def}[Public-coin partition bound]
Let $f \subseteq  \{0,1\}^n \times \Z$ be a relation. Let $\ve > 0$. The $\ve$-public-coin partition bound of $f$, denoted $\prt_\ve(f)$, is given by the optimal value of the following linear program.

\vspace{0.2in}

{\footnotesize
\begin{minipage}{3in}
    \centerline{\underline{Primal}}\vspace{-0.1in}
    \begin{align*}
      \text{min:}\quad & \sum_{z} \sum_{A}  w_{z,A} \cdot 2^{|A|} \\
       \quad &  \forall x  : \sum_{z: (x,z) \in f} \sum_{A: x \in A} w_{z,A} \geq 1 - \epsilon,\\
      & \forall x : \sum_{A: x \in A} \quad \sum_{z}  w_{z,A} = 1 , \\
       & \forall (z,A) : w_{z,A} = \sum_{P: (z,A) \in P} a_p , \\
      & \sum_P a_P = 1  ,\\
      & \forall (z,A)  : w_{z,A} \geq 0 ; \quad  \forall P : a_P \geq 0 \enspace .
    \end{align*}
\end{minipage}
\begin{minipage}{3in}\vspace{-0.5in}
    \centerline{\underline{Dual}}\vspace{-0.1in}
        \begin{align*}
    \text{max:}\quad &   (1-\epsilon)\sum_{x} \mu_{x} + \sum_{x}\phi_{x} + \lambda \\
      &  \forall (z,A) : \sum_{x\in A : (x,z) \in f} \mu_{x}   + \sum_{x\in  A} \phi_{x}  + v_{z,A} \leq 2^{|A|},\\
       &  \forall P : \sum_{(z,A)\in P} v_{z,A}    \geq \lambda ,\\
	& \forall x : \mu_{x} \geq 0, \phi_{x} \in \bR  ; \quad  \forall (z,A) : v_{z,A} \in \bR ,\\
	& \lambda \in \bR \enspace .
    \end{align*}
\end{minipage}
}
\end{Def}

We show that (the logarithm to the base two of) our new bound is a lower bound on randomized query complexity. 
\begin{Lem}
Let $f \subseteq  \{0,1\}^n \times \Z$ be a relation. Let $\ve > 0$. Let $\R_{\ve} (f)$ represent the randomized query complexity of $f$ with worst case error $\ve$. Then,
$$\frac{1}{2} \log_2 \prt_{\ve} (f) \leq \R_{\ve} (f) .$$
\end{Lem}
\begin{proof}
Our proof goes along arguments similar to~\cite{JainK10} for analogous result for the partition bound. 
 
Let $\mcP$ be a randomized query algorithm which achieves $c\defeq \R_\epsilon(f)$. Let $\mcP_r$  be the deterministic query algorithm, arising from $\mcP$, corresponding to random string $r$. We know that each such deterministic query algorithm is a binary decision tree of depth at most $c$ (please refer to~\cite{buhrman:dectreesurvey} for standard definitions related to query complexity). We note that the roots of a decision tree (together) represent a partition of the inputs  into assignments  along with outputs in $Z$. Let $P_r = \{(z^r_1, A^r_1), (z^r_2, A^r_2), \cdots, (z^r_m, A^r_m)  \}$, represent the partition along with outputs corresponding to random string $r$, where $\{A^r_1, \cdots, A^r_m \}$ form a partition of $\{0,1\}^n$  into assignments and $\forall i \in [m], z^r_i \in \Z$. Let $q_r$ represent the probability of string $r$ in $\mcP$. For $P_r$ define $a'_{P_r} \defeq q_r$. For the partitions $P$ that do not correspond to any string $r$ in $\mcP$, define $a'_P=0$.  For any $(z,A)$ define,
$$w'_{z,A} \defeq \sum_{P: (z,A) \in P} a'_P .$$
 
As in the proof of Lemma~\ref{lem:lowercc}, we can argue that $\{w'_{z,A}\} \cup  \{a'_P  \}$ is feasible for the primal of $\prt_\epsilon(f)$.  Note that for each $(z,A)$ with $w'_{z,A} >0$, we have $|A|\leq c$. Also $|P_r| \leq 2^c$ since the depth of the corresponding binary decision tree is at most $c$. Now,
\begin{align*}
\prt_\epsilon(f) &= \sum_{z} \sum_{A}  w'_{z,A} 2^{|A|}  \leq 2^{c} \left(\sum_{z} \sum_{A}  w'_{z,A} \right)  \\
& \leq 2^{c} \left(\sum_r a'_{P_r} \cdot |P_r| \right) \leq 2^{2c} \sum_r a'_{P_r} = 2^{2c} \enspace .
\end{align*}
Hence our result.
\end{proof}

Next we show that the square of  (the logarithm to the base two of)  our new bound forms an upper bound on randomized query complexity.
\begin{Thm}
\label{thm:mainq}
Let $f \subseteq  \{0,1\}^n \times \Z$ be a relation. Let $\ve > 0$. Then,
$$ \R_{2\ve} (f) \leq \left(\log \prt_\ve(f) + \log_2 \frac{1}{\ve} \right)^2 .  $$  
\end{Thm}
\begin{proof}
Let $\prt_\ve(f)  = 2^c $. Let $\{ w_{z,A} \} \cup \{a_P\}$ be an optimal solution for the primal. We have,
$$ \sum_P \sum_{A: (z,A) \in P} a_P \cdot 2^{|A|} = \sum_{z,A} w_{z,A} \cdot 2^{|A|}= 2^c .$$
Define $B \defeq \{ P ~| ~ \exists (z,A) \in P \text{ with } |A| > c + \log_2 \frac{1}{\ve}\}$. Then $\delta \defeq \sum_{P \in B} a_P \leq \ve$. Define $a'_P \defeq \frac{1}{1-\delta} a_P $ for $P\notin B$ and $a'_P \defeq 0$ for $P \in B$. Define $w'_{z,A} \defeq \sum_{P: (z,A) \in P} a'_P$. Then we have,
\begin{align*}
      & \forall x : \sum_{z: (x,z) \in f} \sum_{A: x \in A} w'_{z,A} \geq 1 - 2\epsilon,\\
      & \forall x : \sum_{A: x\in A} \quad \sum_{z}  w'_{z,A} = 1 , \\
      & \forall (z,A) : w'_{z,A} = \sum_{P: (z,A) \in P} a'_P , \\
      & \sum_P a'_P = 1  ,\\
      & \forall (z,A)  : w'_{z,A} \geq 0 ; \quad \forall P : a'_P \geq 0 \enspace .
\end{align*}

We know that a partition with assignments each of length at most $m$ can be realized by a query protocol with  $m^2 $ queries (arguments as in the proof of Theorem 11 of~\cite{buhrman:dectreesurvey}, we reproduce them in Section~\ref{sec:partq} for completeness). Consider a randomized query protocol $\Pi$ as follows. 
\begin{enumerate}
\item Alice  (randomly) chooses a $P= \{(z_1, A_1), (z_2, A_2), \cdots, (z_s, A_s) \} $ with probability $a_P'$.
\item She queries to realize the partition $\{A_1, A_2, \cdots, A_s\}$ with  $(c + \log_2 \frac{1}{\ve})^2$ queries. 
\item If she ends up with assignment $A_i$, she outouts $z_i$.
\end{enumerate}
It is clear that the worst case queries of the protocol is $(c + \log_2 \frac{1}{\ve})^2$. The condition,
 $$\forall x : \sum_{z: (x,z) \in f} \sum_{A: x \in A} w'_{z,A} \geq 1 - 2\epsilon ,$$
implies that the protocol has worst case error at most $2 \ve$.
Therefore,
$$ \R_{2\ve} (f) \leq \left(\log_2 \prt_\ve(f) + \log_2 \frac{1}{\ve} \right)^2 .  $$  
\end{proof}

\subsection*{Acknowledgment}  The work done is supported by the internal grants of the Center for Quantum Technologies (CQT), Singapore. Part of the work done when N.K.V was visiting CQT.

\newcommand{\etalchar}[1]{$^{#1}$}

\appendix 

\section{Communication protocol to realize a partition}
\label{sec:partcc}
Let $\{R_1, R_2, \ldots, R_m \}$ be a partition of $\X \times \Y$.  Let $x$ and $y$ be the inputs to Alice and Bob respectively.

\begin{enumerate}

\item Alice determines if there exists an $i \in [m]$ such that the row corresponding to $x$ intersects with $R_i$ and the number of rectangles in $\{R_1, R_2, \ldots, R_m \}$ that row intersect with $R_i$ are at most $m/2$. If such $i$ exists she communicates it to Bob using $\lceil \log_2 m \rceil$ bits. They both now consider $\{R_1 \cap R_i, R_2 \cap R_i, \ldots, R_m \cap R_i \}$ as a partition of $(\X \times \Y) \cap R_i$ and repeat. If Alice cannot find any such $i$ she indicates this to Bob by sending $0$. 

\item On receiving $0$ from Alice, Bob determines if there exists a $j \in [m]$ such that  the column corresponding to $y$ intersects with $R_j$ and the number of rectangles in $\{R_1, R_2, \ldots, R_m \}$ that column intersect with $R_j$ are at most $m/2$. If such  $j$ exists he communicates it to Alice using $\lceil \log_2 m \rceil $ bits. They both now consider $\{R_1 \cap R_j, R_2 \cap R_j, \ldots, R_m \cap R_j \}$ as a partition of $(\X \times \Y) \cap R_j$ and repeat.  
\end{enumerate}
We can note that either Alice or Bob must succeed in finding desired $i, j$ respectively since the rectangle that contains $(x,y)$ satisfies the requirements in either 1. or 2. above (since $\{R_1, R_2, \ldots, R_m \}$ is a partition of $\X \times \Y$). Also the communication in each round is at most $\lceil \log_2 m \rceil$ and the number of (non-empty) rectangles surviving after each round reduce by a factor of $2$. Hence the process ends after at most $\lceil \log_2 m \rceil$  rounds. The total communication hence is bounded by $(\lceil \log_2 m \rceil)^2$.

\section{Query protocol to realize a partition}
\label{sec:partq}
Let $\{A_1, A_2, \ldots, A_s \}$ be a partition of $\{0,1\}^n$ such that $|A_i| \leq m$ for each $i \in [s]$.  Let $x$ be the string in the database.

\begin{enumerate}

\item Alice queries the bits of $x$ corresponding to $A_1$. If the bits revealed are consistent with $A_1$ then she considers $A_1$ as desired assignment and stops. 

\item If the bits revealed are not consistent with $A_1$ then we note that one bit is revealed for $A_i$ for all $i \in [s]$ (since  $\{A_1, A_2, \ldots, A_s \}$ is a partition of $\{0,1\}^n$). Hence the size of each $A_i$ (consistent with the bits revealed so far) reduces by at least $1$. Alice considers now the new set of (modified) $A_i$s which are consistent with the bits revealed so far and repeats.     
\end{enumerate}
We note that the number of such rounds is at most $m$ and in each round at most $m$ bits are revealed. Hence the total number of queries is at most $m^2$.

\end{document}